\documentclass[journal,10pt,draftclsnofoot,onecolumn,letter]{IEEEtran}
\ifCLASSINFOpdf
\else
\fi
\usepackage[margin=1.1in]{geometry}

\usepackage[cmex10]{amsmath}
\usepackage{array}
\usepackage{mdwmath}
\usepackage{mdwtab}
\usepackage{eqparbox}
\usepackage[tight,footnotesize]{subfigure}
\usepackage{caption}
\usepackage[caption=false,font=footnotesize]{subfig}
\usepackage{fixltx2e}
\usepackage{stfloats}
\usepackage{url}
\usepackage{cite}
\usepackage{graphics,graphicx, latexsym, amssymb, amscd, psfrag}

\newcommand{\cA}{{\mathcal{A}}}

\newcommand{\cC}{{\mathcal{C}}}

\newcommand{\cE}{{\mathcal{E}}}

\newcommand{\bh}{{\mathbf{h}}}

\newcommand{\cJ}{{\mathcal{J}}}

\newcommand{\cM}{{\mathcal{M}}}

\newcommand{\CN}{{\mathcal{CN}}}

\newcommand{\cU}{{\mathcal{U}}}

\newcommand{\cX}{{\mathcal{X}}}

\newcommand{\eps}{\varepsilon}

\usepackage{verbatim} 
 \usepackage{amsxtra}  
\newcounter{actr}
{\begin{list}{(\alph{actr})}{\usecounter{actr}}}{\end{list}}

\newcounter{ictr}
{\begin{list}{(\roman{ictr})}{\usecounter{ictr}}}{\end{list}}

\newtheorem{remark}{Remark}
\newtheorem{thm}{Theorem}
\newtheorem{lemma}{Lemma}

\newtheorem{corol}{Corollary}

\newenvironment{new-proof}[1]
{{\em Proof  #1: }}%
{ \noindent\qed }
\newcommand{\qed}{\rule[0.1ex]{1.4ex}{1.6ex}}

\DeclareMathAlphabet{\mathbsf}{OT1}{cmss}{bx}{n}
\DeclareMathAlphabet{\mathssf}{OT1}{cmss}{m}{sl}

\DeclareSymbolFont{bsfletters}{OT1}{cmss}{bx}{n}  
\DeclareSymbolFont{ssfletters}{OT1}{cmss}{m}{n}
\DeclareMathSymbol{\bsfGamma}{0}{bsfletters}{'000}
\DeclareMathSymbol{\ssfGamma}{0}{ssfletters}{'000}
\DeclareMathSymbol{\bsfDelta}{0}{bsfletters}{'001}
\DeclareMathSymbol{\ssfDelta}{0}{ssfletters}{'001}
\DeclareMathSymbol{\bsfTheta}{0}{bsfletters}{'002}
\DeclareMathSymbol{\ssfTheta}{0}{ssfletters}{'002}
\DeclareMathSymbol{\bsfLambda}{0}{bsfletters}{'003}
\DeclareMathSymbol{\ssfLambda}{0}{ssfletters}{'003}
\DeclareMathSymbol{\bsfXi}{0}{bsfletters}{'004}
\DeclareMathSymbol{\ssfXi}{0}{ssfletters}{'004}
\DeclareMathSymbol{\bsfPi}{0}{bsfletters}{'005}
\DeclareMathSymbol{\ssfPi}{0}{ssfletters}{'005}
\DeclareMathSymbol{\bsfSigma}{0}{bsfletters}{'006}
\DeclareMathSymbol{\ssfSigma}{0}{ssfletters}{'006}
\DeclareMathSymbol{\bsfUpsilon}{0}{bsfletters}{'007}
\DeclareMathSymbol{\ssfUpsilon}{0}{ssfletters}{'007}
\DeclareMathSymbol{\bsfPhi}{0}{bsfletters}{'010}
\DeclareMathSymbol{\ssfPhi}{0}{ssfletters}{'010}
\DeclareMathSymbol{\bsfPsi}{0}{bsfletters}{'011}
\DeclareMathSymbol{\ssfPsi}{0}{ssfletters}{'011}
\DeclareMathSymbol{\bsfOmega}{0}{bsfletters}{'012}
\DeclareMathSymbol{\ssfOmega}{0}{ssfletters}{'012}

\newcommand{\rvg}{{\mathssf{g}}}	

\newcommand{\rvh}{{\mathssf{h}}}	

\newcommand{\rvbh}{{\mathbsf{h}}}

\newcommand{\rvl}{{\mathssf{l}}}	
\newcommand{\rvm}{{\mathssf{m}}}	

\newcommand{\rvn}{{\mathssf{n}}}	
\newcommand{\rvbn}{{\mathbsf{n}}}

\newcommand{\rvq}{{\mathssf{q}}}	

\newcommand{\rvs}{{\mathssf{s}}}	
\newcommand{\rvbs}{{\mathbsf{s}}}

\newcommand{\rvu}{{\mathssf{u}}}	

\newcommand{\rvv}{{\mathssf{v}}}	

\newcommand{\rvw}{{\mathssf{w}}}	

\newcommand{\rvbw}{{\mathbsf{w}}}

\newcommand{\rvx}{{\mathssf{x}}}	

\newcommand{\rvbx}{{\mathbsf{x}}}

\newcommand{\rvy}{{\mathssf{y}}}	

\newcommand{\rvby}{{\mathbsf{y}}}

\newcommand{\rvz}{{\mathssf{z}}}	%

\newcommand{\rvbz}{{\mathbsf{z}}}

\newcommand{\bbz}{\bar{\rvbz}}
\newcommand{\bbZ}{\bar{\mathsf{\mathbf{Z}}}}
\newcommand{\bsi}{\backslash{i}}
\newcommand{\bby}{\bar{\rvby}}

\newcommand{\rvsu}{\bar{\rvs}}
\newcommand{\rvmu}{\bar{\rvm}}

\begin{document}

\title{Private Broadcasting over Independent Parallel Channels }

\author{Ashish~Khisti~\IEEEmembership{Member,~IEEE,} and Tie Liu~\IEEEmembership{Member,~IEEE}
\thanks{A.~Khisti is with the Department
of Electrical and Computer Engineering, University of Toronto,
Toronto, ON M5S 3G4, Canada (e-mail: akhisti@comm.utoronto.ca).
T.~Liu is with the Department
of Electrical and Computer Engineering, Texas A\&M University, College Station, TX 77843, USA (e-mail: tieliu@tamu.edu).}
\thanks{{A. Khisti's work was supported by an NSERC (Natural Sciences Engineering Research Council) Discovery Grant and an Ontario Early Researcher Award. 
T. Liu was supported by the National Science Foundation under Grants
CCF-08-45848 and CCF-09-16867.}
}\thanks{Part of this work was presented in the IEEE International Symposium on Information Theory (ISIT),
Cambridge, MA, July, 2012}}

\maketitle

\begin{abstract}
We study private broadcasting of two messages to two groups of receivers over independent parallel channels. 
One group consists of an arbitrary number of receivers interested in a common message, whereas the other group has only one receiver. 
Each message must be kept confidential from the receiver(s) in the other group.
Each of the sub-channels is degraded, but the order of receivers on each channel can be different.  
While corner points of the capacity region were characterized in earlier works,
we establish the capacity region and show the optimality of a superposition strategy. 
For the case of parallel Gaussian channels, we show that a Gaussian input distribution is optimal.
We also discuss an extension of our setup to  broadcasting over a block-fading channel
and demonstrate significant performance gains  using the proposed scheme
over a baseline time-sharing scheme.
\end{abstract}

\IEEEpeerreviewmaketitle

\section{Introduction}
There has been a considerable amount of interest in recent years in exploiting the properties of fading  wireless channels for transmission of confidential messages (see e.g.,~\cite{gopalaLaiElGamal:06Secrecy, yates, liang, ashishBC, barros,Liang:09Compound} and references therein). Such studies have lead to new coding techniques such as 
the variable rate extension of the wiretap codebook~\cite{gopalaLaiElGamal:06Secrecy}, secure product codebooks~\cite{LiuPrabhakaran:08} and
secure multicast codebooks~\cite{ashishBC}. In the present work we study a setup where a single transmitter needs to serve two groups of receivers over a block-fading  channel. There are $K$ receivers in group $1$, all interested in a common message, whereas there is  a single receiver in group $2$. The message of group $1$ must be kept confidential from the group $2$ receiver, whereas the message of group $2$ must be kept confidential from group $1$.  We will refer to this setup as {\em{private broadcasting}}.  In related work, references~\cite{caiLam00, Liu09, LiuMaricSpasojevicYates:07} study private broadcasting when there is one receiver in each group. 
References~\cite{diggavi:12,yang:11} study private broadcasting with feedback over erasure and MIMO broadcast channels.
Reference~\cite{ashishCompound} studies interference alignment techniques for private broadcasting. In this paper we focus on the case when there are $M$ independent, parallel and degraded sub-channels and thereafter treat the natural extension to block-fading channels.   

Our setup reduces to previously known results at the corner points of the capacity region. When we only need to transmit the message for group $1$, with the group $2$ receiver as the only eavesdropper, the capacity can be achieved using a secure multicast codebook~\cite{ashishBC}. Instead, when we only need to transmit the message for group $2,$ with all receivers in group $1$ as eavesdroppers, the capacity can be achieved using a secure product codebook \cite{LiuPrabhakaran:08}. Interestingly the secure multicast  and  secure product  codebook constructions are based on  different ideas. A secure multicast codebook consists of $M$ sub-codebooks, one for each channel. Each sub-codebook is a wiretap codebook~\cite{csiszarKorner:78},  has the same rate as the transmitted message and guarantees confidentiality of the message from the eavesdropper on its respective link. The secure multicast construction  guarantees that the legitimate receiver can decode the message by using the output of all the channels. Furthermore the message remains confidential from the eavesdropper even when all the channel outputs are combined. While the secure product codebook also uses one
sub-codebook for each sub-channel, the rate of each sub-codebook equals the capacity of the legitimate receiver on that sub-channel. The secure product codebook takes a cartesian product of these codebooks and then applies the wiretap construction to this product codebook.
This guarantees that the output codeword on any given sub-channel is (nearly) independent of the output codewords on other sub-channels. This limits the amount of information that gets leaked to an eavesdropper on any given sub-channel. Both the secure multicast codebook and secure product codebook result in a higher rate than a vector extension of  the wiretap codebook to parallel channels.

In this paper we  study the case when both the messages need to be simultaneously transmitted. We find that a superposition construction achieves the entire capacity region. The proposed construction imposes a particular layering order for the secure multicast and secure product codebooks. The codewords in each sub-codebook of the secure product codebook must constitute the cloud centers, whereas the codewords in the associated sub-codebook of the secure multicast codebook must constitute satellite codewords. The optimality of such a layered coding scheme was somewhat unexpected. 
In absence of secrecy constraints, to the best of our knowledge the capacity region in the proposed setup remains open, even though the corner points are known~\cite{elGamal:80}.  We will provide an explanation on the sufficiency of the superposition approach after presenting the coding scheme in section~\ref{sec:coding}. 

For the case of independent Gaussian sub-channels, we further establish that a Gaussian input distribution is optimal. The proof involves obtaining a Lagrangian dual for every boundary point  of the capacity region and then using an extremal inequality~\cite{costa,liu-vector} to show that the expression is maximized using Gaussian inputs. The result for the Gaussian channels are extended to a block-fading channel model using suitable quantization of the channel gains. We numerically evaluate the rate region for a sub-optimal power allocation and observe significant gains over a naive time-sharing approach.

\section{Problem Statement and Main Results}

\subsection{Independent Parallel Channels}

Our setup involves $M$ independent parallel sub-channels and two groups  of receivers. There are $K$ receivers in group $1$ and one  receiver in group $2$. The output symbols at receiver $k$ in group $1$ across the $M$ sub-channels is denoted by
\begin{align}
\rvby_k  = (\rvy_{k,1}, \rvy_{k,2},\ldots, \rvy_{k,M}), \quad k=1,2\ldots, K,
\end{align}
whereas the output symbols of the group $2$ receiver across the $M$ sub-channels are denoted by 
\begin{align}
\rvbz = (\rvz_1,\rvz_2,\ldots, \rvz_M),
\end{align} 
and the channel input symbols are denoted by $\rvbx = (\rvx_1,\ldots, \rvx_M)$.  
 
Each sub-channel is a degraded broadcast channel. The degradation on sub-channel $i$ can be expressed as
\begin{align}
\rvx_i \rightarrow \rvy_{\pi_i(1),i} \cdots \rvy_{\pi_i(l_i),i} \rightarrow \rvz_i \rightarrow \rvy_{\pi_i(l_i+1),i} \cdots \rvy_{\pi_i(K),i},
\label{eq:degOrd} 
\end{align}
for some permutation $\{\pi_i(1),\ldots, \pi_i(K)\}$ of the set $\{1,\ldots, K\}$. 

We intend to transmit message $\rvm_1$ to receivers $1,\ldots, K$  in group $1,$ while the message $\rvm_2$ must be transmitted to the receiver in group $2$. A length-$n$ private broadcast code encodes a message pair $(\rvm_1, \rvm_2) \in [1,2^{nR_1}] \times [1,2^{nR_2}]$  into a sequence $\rvbx^n$ such that  ${\Pr(\rvm_1 \neq \hat{\rvm}_{1,k}) \le \eps_n},$ and ${\Pr(\rvm_2 \neq
\hat{\rvm}_2) \le \eps_n},$ and furthermore the secrecy constraints
\begin{align}
\frac{1}{n} I(\rvm_1; \rvbz^n) \le \eps_n, \quad \frac{1}{n} I(\rvm_2;\rvby_{k}^n) \le \eps_n, \quad k=1,2,\ldots, K,
\label{eq:secrecy}\end{align} 
are also satisfied. Here $\{\eps_n\}$  approaches zero as $n\rightarrow\infty$. The capacity region consists of the set of all  rate pairs $(R_1, R_2)$ achieved by some private broadcast code. The following Theorem characterizes this region.

\begin{thm}
\label{thm:DMC}
Let auxiliary variables $\{\rvu_i\}_{1\le i \le M}$  satisfy the  Markov condition
\begin{align}
\rvu_i \rightarrow \rvx_i \rightarrow \rvy_{\pi_i(1),i} \cdots \rvy_{\pi_i(l_i),i} \rightarrow \rvz_i \rightarrow \rvy_{\pi_i(l_i+1),i} \cdots \rvy_{\pi_i(K),i}.
\label{eq:r1}
\end{align}
The capacity region is given by the union of all rate pairs  $(R_1,R_2)$ that satisfy the following constraints:
\begin{align}
R_1 &\le \min_{1\le k \le K}\left\{\sum_{i=1}^M I(\rvx_i; \rvy_{k,i}| \rvu_i, \rvz_i)\right\}\label{eq:R1e}\\
R_2 &\le \min_{1\le k \le K}\left\{\sum_{i=1}^M I(\rvu_i; \rvz_i | \rvy_{k,i})\right\}\label{eq:R2e}
\end{align}
for some choice of $\{\rvu_i\}_{1\le i \le M}$ that satisfy~\eqref{eq:r1}. The alphabet of $\rvu_i$ satisfies the cardinality constraint $|\cU_i| \le |\cX_i| + 2K-1$. \hfill$\Box$
\end{thm}

The coding theorem and converse for Theorem~\ref{thm:DMC} are presented in section~\ref{sec:coding} and~\ref{sec:conv} respectively.

\subsection{Gaussian Channels}
Consider the discrete-time real Gaussian model where the channel output over sub-channel $i$ at time index $t$ is given by
\begin{align}
\rvy_{k,i}(t) &= \rvx_{i}(t) + \rvn_{k,i}(t)\\
\rvz_{i}(t) &= \rvx_{i}(t) + \rvw_{i}(t), \quad t=1,\ldots,T.\label{eq:Gau}
\end{align}
The additive noise vectors $\rvbn_{k,i}=(\rvn_{k,i}(1),\cdots,\rvn_{k,i}(T))$ and $\rvbw_i=(\rvw_{i}(1),\cdots,\rvw_{i}(T))$ have entries that are sampled  i.i.d.\ $\mathcal{N}(0,\sigma_{k,i}^2)$ and $\mathcal{N}(0,\delta_i^2)$, respectively. Since the capacity region of the channel depends on the joint distribution of the additive noise $(\rvn_{1,i}(t),\ldots,\rvn_{K,i}(t),\rvw_{i}(t))$ only through the marginals and that Gaussian variables are infinitely divisible, without loss of generality we may assume that for each sub-channel $i$ the receivers are degraded as expressed in \eqref{eq:degOrd}. We shall consider both the per sub-channel average power constraint 
\begin{align}
\frac{1}{T}E\left[\|\rvbx_i\|^2\right] & \leq P_i, \quad \forall i=1,\ldots,M\label{eq:PAPC}
\end{align}
and the total average power constraint
\begin{align}
\frac{1}{T}\sum_{i=1}^{M}E\left[\|\rvbx_i\|^2\right] & \leq P\label{eq:TAPC}
\end{align}
where $\rvbx_{i}=(\rvx_{i}(1),\cdots,\rvx_{i}(T))$ is the input vector for sub-channel $i$. 

\begin{thm}
\label{thm:Gauss}
The capacity region under the per sub-channel average power constraint \eqref{eq:PAPC} is given by the union of all rate pairs $(R_1,R_2)$ that satisfy the following constraints:
\begin{align}
R_1 &\leq \min_{1 \leq k \leq K} \left\{\sum_{i=1}^{M}A^{(1)}_{k,i}(\mathbf{Q})\right\}\label{eq:Gau1}\\
R_2 &\leq \min_{1 \leq k \leq K} \left\{\sum_{i=1}^{M}A^{(2)}_{k,i}(\mathbf{Q})\right\}\label{eq:Gau2}
\end{align}
for some  power vector $\mathbf{Q}=(Q_1,\ldots,Q_M)$, where $0 \leq Q_i \leq P_i$ for all $i=1,\ldots,M$, 
\begin{align}
A^{(1)}_{k,i}&(\mathbf{Q})
:= \left[\frac{1}{2}\log\left(\frac{Q_i+\sigma_{k,i}^2}{\sigma_{k,i}^2}\right)-
\frac{1}{2}\log\left(\frac{Q_i+\delta_{i}^2}{\delta_{i}^2}\right)\right]^+\\
A^{(2)}_{k,i}&(\mathbf{Q})
:= \left[\frac{1}{2}\log\left(\frac{P_i+\delta_{i}^2}{Q_i+\delta_{i}^2}\right)-
\frac{1}{2}\log\left(\frac{P_i+\sigma_{k,i}^2}{Q_i+\sigma_{k,i}^2}\right)\right]^+
\end{align}
and $x^+:=\max\{x,0\}$.\hfill$\Box$
\end{thm}

A proof of Theorem~\ref{thm:Gauss} is provided in section~\ref{sec:Gau}.

\begin{corol}
\label{corol:Gauss}
The capacity region under the total average power constraint \eqref{eq:TAPC} is given by the union of all rate pairs $(R_1,R_2)$ that satisfy the constraints \eqref{eq:Gau1} and \eqref{eq:Gau2} for some power vectors $\mathbf{P}=(P_1,\ldots,P_M)$ and $\mathbf{Q}=(Q_1,\ldots,Q_M)$, where $0 \leq Q_i \leq P_i$ for all $i=1,\ldots,M$ and $\sum_{i=1}^{M}P_i \leq P$.\hfill$\Box$
\end{corol}

The above corollary follows directly from Theorem~\ref{thm:Gauss} and the well-known connection between the per sub-channel and the total average power constraints. We will not provide a proof of Corollary~\ref{corol:Gauss}.

\subsection{Fading Channels}
We consider a block-fading channel model with a coherence period of $T$ complex symbols. The channel output in coherence block $i$ is given by
\begin{align}
\rvby_{k}(i) &= \rvh_k(i) \rvbx(i) + \rvbn_k(i)\\
\rvbz(i) &= \rvg(i) \rvbx(i) + \rvbw(i), \quad i=1,2\ldots, M\label{eq:fading}
\end{align}
where the channel gains $\rvh_k(i)$ of the $K$ receivers in group $1$  and the channel gain $\rvg(i)$ of the group $2$ receiver are sampled independently in each coherence block $i$ and stay constant throughout the block. The coherence period $T$ will be taken to be sufficiently large so that random coding arguments can be invoke in each coherence block. The channel input $\rvbx(i) \in {\mathbb{C}}^T$ satisfies a long-term average power constraint 
\begin{align}
E\left[\frac{1}{MT}\sum_{i=1}^M||\rvbx(i)||^2\right] & \le P
\end{align}
whereas the additive noise vectors $\rvbn_k(i)$ and $\rvbw(i)$ have entries that are sampled  i.i.d.\ $\CN(0,1)$. We are interested in the ergodic communication scenario where the number of blocks $M$ used for communication can be arbitrarily large. Furthermore we assume that the channel gains in each coherence block are revealed to all terminals including the transmitter at the beginning of each coherence block.

\begin{thm}
The private broadcasting capacity region for the fading channel model consists of all rate pairs  $(R_1, R_2)$ that satisfy the following 
constraints:
\begin{align}
R_1 &\le \min_{1\le k \le K}E\left[\left\{\log\left(\frac{1+ Q(\rvbh, \rvg)|\rvh_k|^2}{1+Q(\rvbh, \rvg)|\rvg|^2}\right)\right\}^+\right],\label{eq:R1-fading}\\
R_2 &\le \min_{1\le k \le K} E\bigg[\left\{\log\left(\frac{1 + P(\rvbh, \rvg) |\rvg|^2}{1+Q(\rvbh,\rvg)|\rvg|^2}\right) - \log\left( \frac{1 + P(\rvbh, \rvg) |\rvh_k|^2}{1+Q(\rvbh,\rvg)|\rvh_k|^2}\right)\right\}^+\bigg],\label{eq:R2-fading}
\end{align}
for some power allocation functions $P(\rvbh, \rvg)$ and $Q(\rvbh, \rvg)$ that satisfy $0 \le Q(\rvbh, \rvg) \le P(\rvbh, \rvg)$ for all $(\rvbh, \rvg) \in {\mathbb C}^{K+1},$ and $E[P(\rvbh, \rvg)] \le P$, where $\rvbh := (\rvh_1, \ldots, \rvh_K)$ denotes the channel gains of the receivers in group $1$. \hfill$\Box$
\label{thm:fading}
\end{thm}

A proof of Theorem~\ref{thm:fading} is provided in Section~\ref{sec:fading}.

Theorems~\ref{thm:DMC},~\ref{thm:Gauss}  and~\ref{thm:fading} constitute the main results in this paper. 

\section{Coding Theorem}
\label{sec:coding}
\begin{figure}
\centering
\includegraphics[scale=0.45, angle = 270]{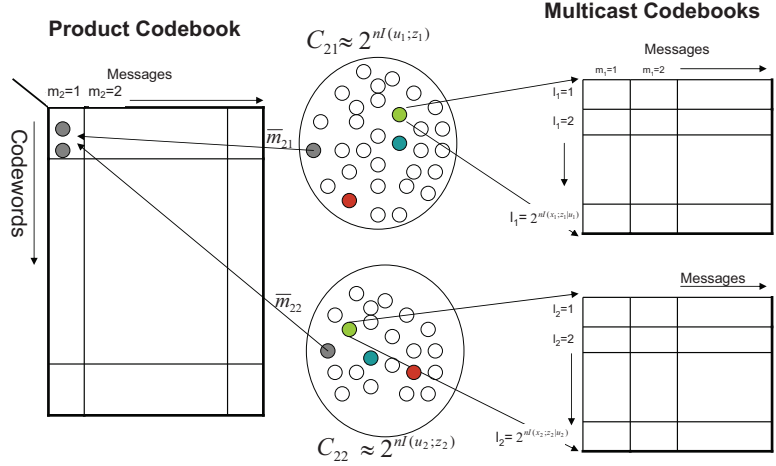}

\vspace{-3em}
\caption{Superposition construction for the case of two channels. The product codebook for the group $2$ user is obtained by taking a cartesian
$\cC_{21} \times \cC_{22}$ of two independently generated codebooks and binning the resulting codeword pairs. The multicast codebook is generated, conditioned on the codewords of $\cC_{21}$ and $\cC_{22}$.}
\label{fig:superposition}
\end{figure}

The basic idea behind our coding scheme is illustrated in Fig.~\ref{fig:superposition}.  The message $\rvm_2$ is encoded using a product codebook~\cite{LiuPrabhakaran:08,gopalaLaiElGamal:06Secrecy}, whose codewords are obtained by taking cartesian product of the $M$ codebooks, one for each of the parallel channels. The message $\rvm_1$ is encoded using a multicast codebook~\cite{ashishBC}, also consisting of $M$ codebooks. As shown in Fig.~\ref{fig:superposition}, the codewords of the product-codebook constitute cloud centers
of the superposition codebook, whereas the codewords of the multicast codebook constitute the satellite codewords. We describe the details of our construction  in the following sub-sections.

\subsection{Product-Codebook Construction}
The message $\rvm_2$ is encoded using a product codebook~\cite{LiuPrabhakaran:08,gopalaLaiElGamal:06Secrecy}. Let $\cM_{2,i}$ be the set of all binary sequences of length $N_{2,i}=n(I(\rvu_i; \rvz_i)-2\eps)$ i.e.,
\begin{align}
\cM_{2,i} := \{0,1\}^{N_{2,i}}.
\label{eq:m-2i}
\end{align}
On channel $i$, we generate a codebook $\cC_{2,i}: \cM_{2,i} \rightarrow \cU_{i}^n$ consisting of $|\cM_{2,i}|$ codewords, i.e.,
\begin{align}
\cC_{2,i} := \left\{u_i^n(\rvmu_{2,i}): \rvmu_{2,i} \in \left[1, 2^{N_{2,i}}\right]\right\},
\end{align}
where each sequence $u_i^n$ is sampled i.i.d.\ from the distribution $p_{\rvu_i}(\cdot)$. Let 
 \begin{align}
\cM_2 &:= \cM_{2,1} \times \cM_{2,2} \times \ldots \times \cM_{2,M}\\
\label{eq:M-prod}
&=\left\{(\rvmu_{2,1},\ldots, \rvmu_{2,M}): \rvmu_{2,i} \in \{0,1\}^{N_i}, i=1,\ldots, M\right\}.
\end{align}

As shown in Fig.~\ref{fig:superposition}, we partition the set $\cM_2$ into $2^{nR_2}$ bins such that there are $L_2 = 2^{n\left\{\sum_{i=1}^M I(\rvu_i;\rvz_i) -R_2-M\eps \right\}}$ sequences in each bin. 
Each bin corresponds to one message $\rvm_2 \in [1,2^{nR_2}]$. 
Thus given a message $\rvm_2$ the encoder selects one sequence   $(\rvmu_{2,1},\ldots, \rvmu_{2,M})\in \cM_2$ uniformly at random from the corresponding bin.
On channel $i$ we select the codeword  $\rvu_i^n \in \cC_{2,i}$ associated with $\rvmu_{2,i}$. We note that from our construction, each sequence in $\cM_2$ is equally likely i.e.,
\begin{multline}
\Pr(\rvmu_{2,1}=\bar{m}_{2,1},\ldots, \rvmu_{2,M}= \bar{m}_{2,M})  = \prod_{j=1}^M \Pr(\rvmu_{2,j}=\bar{m}_{2,j}) = \frac{1}{|\cM_{2,1}|\times |\cM_{2,2}|\ldots, |\cM_{2,M}|}.\label{eq:M2-exp}
\end{multline}

\subsection{Multicast-Code Construction}
The codebook associated with $\rvm_1$ is a secure multicast codebook~\cite{ashishBC}. 
For each $\rvu_i^n \in \cC_{2,i},$ and each $\rvm_1 \in [1,2^{nR_1}]$ we construct a codebook
$\cC_{1,i}(\rvu_i^n, m_1)$ consisting of a total of  $L_{1,i}=2^{n (I(\rvx_i; \rvz_i|\rvu_i)+\eps)}$ codeword sequences of length $n$, 
each sampled i.i.d. from the distribution $\prod_{j=1}^n p_{\rvx_i|\rvu_i}(x_{ij}|u_{ij})$. 

Let $\rvl_{1,i}$ be uniformly distributed over $[1,L_{1,i}]$.  Given a message $\rvm_1 \in [1,2^{nR_1}]$ and codewords $(\rvu_1^n,\ldots, \rvu_M^n),$ 
selected in the base layer, we select the sequence $\rvx_i^n$  from the codebook $\cC_{1,i}(\rvu_i^n, \rvm_1)$ corresponding to  the
randomly and uniformly generated index $\rvl_{1,i}$. The sequence $\rvx_i^n$ is transmitted on sub-channel $i$.

The following property will be useful in our subsequent analysis. 
\begin{lemma}
The sequences $(\rvx_1^n, \rvx_2^n,\ldots, \rvx_M^n)$ are conditionally independent given  $\rvm_1$.
\label{lem:cond-ind}
\end{lemma}
\begin{proof}
Note that
\begin{align}
p(\rvx_1^n,\ldots, \rvx_M^n|\rvm_1) &= \sum_{\{\rvmu_{2,i}\}} p(\rvx_1^n,\ldots, \rvx_M^n, \rvmu_{2,1},\ldots, \rvmu_{2,M}|\rvm_1)\\
&=\sum_{\{\rvmu_{2,i}\}} p(\rvx_1^n,\ldots, \rvx_M^n |\rvm_1,\rvmu_{2,1},\ldots, \rvmu_{2,M})p(\rvmu_{2,1},\ldots, \rvmu_{2,M})\label{eq:msg-indep}\\
&=\sum_{\{\rvmu_{2,i}\}} p(\rvx_1^n,\ldots, \rvx_M^n |\rvm_1,\rvmu_{2,1},\ldots, \rvmu_{2,M})p(\rvmu_{2,1})\ldots p(\rvmu_{2,M})\label{eq:msg-indep2}\\
&=\sum_{\{\rvmu_{2,i}\}} p(\rvx_1^n|\rvm_1, \rvmu_{2,1})\ldots p(\rvx_M^n|\rvm_1,\rvmu_{2,M})p(\rvmu_{2,1})\ldots p(\rvmu_{2,M})\label{eq:xindep}\\
&=\prod_{i=1}^M \sum_{\rvmu_{2,i}} p(\rvx_i^n|\rvm_1, \rvmu_{2,i})p(\rvmu_{2,i})\\
&=\prod_{i=1}^M \sum_{\rvmu_{2,i}} p(\rvx_i^n,\rvmu_{2,i}|\rvm_1)\\
&=\prod_{i=1}^M p(\rvx_i^n|\rvm_1)\label{eq:xindep2}
\end{align}
where~\eqref{eq:msg-indep} follows from the fact that the messages $\rvmu_{2,1},\ldots, \rvmu_{2,M}$ are independent of $\rvm_1$;~\eqref{eq:msg-indep2} follows from the fact that the messages satisfy~\eqref{eq:M2-exp};~\eqref{eq:xindep} follows from the fact that each $\rvx_i^n \in \cC_{1,i}(\rvm_1, \rvu_i^n)$ and $\rvu_i^n$ is a function of $\rvmu_{2,i}$. Eq.~\eqref{eq:xindep2} establishes the conditional independence of the messages and completes the proof.
\end{proof}

\subsection{Decoding and Error Analysis}
\subsubsection{Decoding of Message $\rvm_1$}
Receiver $k$ in group $1$ selects those sub-channels $\cJ_k$ where it is stronger than the group $2$ receiver:
\begin{align}
\cJ_k = \big\{i \in [1,M]: \rvx_i \rightarrow \rvy_{k,i} \rightarrow \rvz_{i} \big\} \label{eq:j-def}
\end{align}
\begin{itemize}
\item For each $i \in \cJ_k$, receiver $k$ selects a sequence $\hat{\rvu}^n_i \in \cC_{2,i}$ such that\footnote{We will use the notion of strong typicality. The set $T_\eps^n(\rvx, \rvy)$ denotes the $\eps$-strongly typical set. } $(\hat{\rvu}_i^n, \rvy_{k,i}^n) \in T_\eps^n(\rvu_i,\rvy_{k,i})$. We define $\cE_k$ as the event that there exists some $i \in \cJ_k$ such that $\{\hat{\rvu}_i^n \neq \rvu_i^n\}$.
\item Receiver $k$ then searches for a message $\hat{\rvm}_1 \in [1,2^{nR_1}]$ with the following property: for each $i \in \cJ_k$ there exists a codeword $\rvx_i^n \in \cC_{1,i}(\rvm_1, \hat{\rvu}_i^n)$ such that $(\rvx_i^n, \rvy_{k,i}^n) \in T_\eps^n(\rvx_i, \rvy_{k,i}|\rvu_i)$.  An error is declared if $\hat{\rvm}_1 \neq \rvm_1$. 
\end{itemize}
Now observe that
\begin{align}
\Pr(\hat{\rvm}_1 \neq \rvm_1) &\le \Pr(\cE_k)  + \Pr(\hat{\rvm_1}\neq \rvm_1| \cE_k^c).\label{eq:pe-anal}
\end{align}

Since $\left|{\cC_{2,i}}\right| \le 2^{n(I(\rvu_i;\rvz_i)-\eps)}$ and $I(\rvu_i;\rvy_{k,i})\ge I(\rvu_i;\rvz_i)$ for each $i \in \cJ_k,$ it follows that
$\Pr(\cE_k) \le M\eps$.

To bound the second term in~\eqref{eq:pe-anal} we use the union bound and analysis of typical events.
\begin{align}
\Pr(\hat{\rvm}_1\neq \rvm_1 | \cE_k^c) &\le 2^{nR_1}\prod_{i \in \cJ_k}\big\{\left|\cC_{1,i}\right|2^{-n(I(\rvx_i; \rvy_{k,i}|\rvu_i) -\eps)}\big\}\\
&\le 2^{nR_1}2^{-n \sum_{i\in \cJ_k}\left(I(\rvx_i;\rvy_{k,i}|\rvu_i) - I(\rvx_i;\rvz_i|\rvu_i)-2\eps\right) }\\
&= 2^{nR_1} 2^{-n \sum_{i\in \cJ_k}\left(I(\rvx_i;\rvy_{k,i}|\rvu_i,\rvz_i) - 2\eps\right)}
\end{align}
which goes to zero provided that  $R_1 \le \sum_{i\in \cJ_k} I(\rvx_i;\rvy_{k,i}|\rvu_i,\rvz_i) -(2M+1)\eps $. Since $\eps>0$ is arbitrary, our choice of $R_1$ in~\eqref{eq:R1e} thus guarantees that the error probability associated with message $\rvm_1$ vanishes to zero.

\subsubsection{Decoding of message $\rvm_2$}
The receiver in group $2$ decodes message $\rvmu_{2,i}$ on sub-channel $i$ by searching for a sequence $\rvu_i^n \in \cC_{2,i}$ that is jointly typical with $\rvz_i^n$. Since the number of codewords in $\cC_{2,i}$ does not exceed $2^{n (I(\rvu_i;\rvz_i)-2\eps)},$ this event succeeds with high probability.  Hence the receiver correctly decodes $(\rvmu_{2,1},\ldots, \rvmu_{2,M})$ and in turn message $\rvm_2$ with high probability.

\subsection{Secrecy Analysis}
In order to establish the secrecy of message $\rvm_1$ we need to show that
\begin{equation}
\frac{1}{n}I(\rvm_1; \rvbz^n | \cC) \le \eps_n \label{eq:m1-sec}
\end{equation}

Using Lemma~\ref{lem:cond-ind} and the fact that the channels are independent, we have that $\rvz_1^n,\ldots, \rvz_M^n$ are conditionally independent
given $\rvm_1$. It follows that
\begin{align}
\frac{1}{n}I(\rvm_1; \rvbz^n | \cC) \le  \sum_{i=1}^M I(\rvm_1; \rvz_i^n |\cC).\label{eq:m1-sec-sum}
\end{align}
Since in our conditional codebook construction, there are $2^{n (I(\rvx_i; \rvz_i |\rvu_i)+\eps)}$ sequences in each codebook $\cC_{1,i}(\rvu_i^n,\rvm_1),$
it follows from standard arguments that $\frac{1}{n}I(\rvm_1; \rvz_i^n|\cC) \le \eps_n$. The secrecy constraint~\eqref{eq:m1-sec}  now follows.

To establish secrecy of message $\rvm_2$ with respect to  user $1$ in group $1$, we show that
\begin{equation}
\frac{1}{n}H(\rvm_2 | \rvby_1^n, \rvm_1) \ge R_2 - \eps_n. \label{eq:sec-cond}
\end{equation}
where for simplicity we drop the subscript associated with user $1$ in the sequence $\rvby_1^n$.
Without loss of generality, we assume that sub-channels $i=1,2,\ldots, L$ satisfy $\rvx_i \rightarrow \rvz_{i} \rightarrow \rvy_{i}$
while sub-channels $i = L+1,\ldots, M$ satisfy $\rvx_i \rightarrow \rvy_i \rightarrow \rvz_{i}$.  Now consider
\allowdisplaybreaks{\begin{align}
H(\rvm_2 | \rvby_1^n, \rvm_1) &= H(\rvm_2|\rvy_1^n,\ldots, \rvy_M^n,\rvm_1)\\
&=H(\rvmu_{2,1}^M | \rvy_1^n, \ldots, \rvy_M^n, \rvm_1)  - H(\rvmu_{2,1}^M|\rvm_2, \rvm_1, \rvy_1^n,\ldots, \rvy_M^n)\\
&=\sum_{j=1}^M H(\rvmu_{2,j}|\rvy_j^n, \rvm_1) -H(\rvmu_{2,1}^M|\rvm_2, \rvm_1, \rvy_1^n,\ldots, \rvy_M^n)\label{eq:m-indep}\\
&\ge\sum_{j=1}^L H(\rvmu_{2,j}|\rvy_j^n, \rvm_1) -H(\rvmu_{2,1}^M|\rvm_2, \rvm_1, \rvy_1^n,\ldots, \rvy_M^n)\label{eq:ent-bnd2}
\end{align}}where~\eqref{eq:m-indep} follows by establishing that the collection of pairs $\{(\rvm_{2,1}, \rvy_1^n),\ldots, (\rvm_{2,M},\rvy_M^n)\}$   is conditionally independent given $\rvm_1$, which can be establishes in a manner similar to the proof of Lemma~\ref{lem:cond-ind} and~\eqref{eq:ent-bnd2}
follows from the fact that the entropy function is non-negative and therefore we can drop the terms $L+1,\ldots, M$ in the first summation.

We lower bound the first term in~\eqref{eq:ent-bnd2}. Recall that $\rvmu_{2,j}$ is uniformly distributed over $\cC_{2,j}$ with $|\cC_{2,j}| = 2^{n(I(\rvu_j;\rvz_j)-\eps)}$. Furthermore, the corresponding codeword $\rvu_j^n$ is the base codeword in $\cC_{1,j}(\rvm_1, \rvu_j^n)$ and$$\left|\cC_{1,j}(\rvm_1, \rvu_j^n)\right|= 2^{n(I(\rvx_j;\rvz_j|\rvu_j)-\eps)} \ge 2^{n(I(\rvx_j;\rvy_j|\rvu_j)-\eps)},$$ since the channel satisfies the relation $\rvx_j\rightarrow\rvz_j\rightarrow\rvy_j$ for $j=1,\ldots, L$. Since the satellite codeword $\rvx_j^n$ is uniformly selected from $\cC_{1,j}$ it follows that~\cite[Remark 22.2, pp.~554-555]{ElGamal-Kim}
\begin{align}
\frac{1}{n}H(\rvmu_{2,j}|\rvy_j^n, \rvm_1) \ge I(\rvu_j;\rvz_j) -I(\rvu_j;\rvy_j) -\eps. \label{eq:term-a0}
\end{align}
and therefore using the fact that $\rvu_j \rightarrow  \rvz_j \rightarrow \rvy_j,$ we have
\begin{align}
\frac{1}{n}\sum_{j=1}^LH(\rvmu_{2,j}|\rvy_j^n, \rvm_1) \ge \sum_{j=1}^LI(\rvu_j;\rvz_j|\rvy_j)  -L\eps. \label{eq:term-a}
\end{align}
We next upper bound the second term in~\eqref{eq:ent-bnd2}. Note that
\begin{align}
H(\rvmu_{2,1}^M|\rvm_2, \rvm_1, \rvy_1^n,\ldots, \rvy_M^n) \le H(\rvmu_{2,1}^M|\rvm_2, \rvm_1, \rvy_1^n,\ldots, \rvy_L^n,\rvz_{L+1}^n,\ldots \rvz_M^n)\label{eq:degradation} 
\end{align} since $\rvz_j^n$ is a degraded version of $\rvy_j^n$ on channels $j \in \{L+1,\ldots, M\}$. Also note that
\begin{align}
\tilde{R} &= \frac{1}{n}H(\rvmu_{2,1},\ldots, \rvmu_{2,M})\\
&=\frac{1}{n}\sum_{i=1}^M H(\rvmu_{2,i})\\
&= \sum_{i=1}^M \left\{I(\rvu_i; \rvz_i) - 2\eps\right\},\label{eq:Rtexp}
\end{align}
where we use the fact that the messages $(\rvmu_{2,1},\ldots, \rvmu_{2,M})$ are mutually independent (c.f.~\eqref{eq:M2-exp}). Furthermore we select
\begin{align}
R_2 &= \frac{1}{n}H(\rvm_2)\\
&\le \sum_{i=1}^L I(\rvu_i; \rvz_i| \rvy_i) -(2M+1)\eps.\label{eq:R2exp}
\end{align}
Note that
\begin{align}
\tilde{R} - R_2 &> \sum_{i=1}^L I(\rvu_i; \rvy_i) + \sum_{i=L+1}^M I(\rvu_i; \rvz_i)\\ &= I(\rvu_1,\ldots, \rvu_M; \rvy_1,\ldots, \rvy_L, \rvz_{L+1},\ldots, \rvz_M)
\end{align}
where the last step follows from the fact that we have selected $\rvu_1,\ldots, \rvu_M$ to be mutually independent and the channels are also independent.
We can therefore conclude that (c.f.~\cite[Lemma 22.1,~Remark 22.2, pp.~554-555]{ElGamal-Kim},\cite[Lemma~1]{khiang:12})
\begin{align}
\frac{1}{n}&H(\rvmu_{2,1}^M|\rvm_2, \rvm_1, \rvy_1^n,\ldots, \rvy_L^n,\rvz_{L+1}^n,\ldots \rvz_M^n) \notag\\
&\le \tilde{R} - R_2 - I(\rvu_1,\ldots, \rvu_M; \rvy_1,\ldots, \rvy_L, \rvz_{L+1},\ldots, \rvz_M) + \eps\\
&= \sum_{i=1}^L I(\rvu_i; \rvz_i|\rvy_i) - R_2 + \eps \label{eq:term-b}. 
\end{align}
Substituting~\eqref{eq:term-a} and~\eqref{eq:term-b} into~\eqref{eq:ent-bnd2} we have that
\begin{align}
\frac{1}{n}H(\rvm_2 | \rvby_1^n, \rvm_1)  &\ge R_2-(L+1)\eps,
\end{align}
Since $\eps >0$ can be arbitrarily small, this establishes the secrecy of message $\rvm_2$ with respect to user $1$ in group $1$. The secrecy with respect to every other
user can be established in a similar fashion.

\begin{remark}
The superposition approach uses the codewords for the group $2$ user as cloud centers and the codewords of the group $1$ user as satellite codewords.
To justify this, note that on any given channel, say channel $i$, there is an ordering of receivers  as in~\eqref{eq:degOrd}. Receivers $\{\pi_i(l_i+1),\ldots, \pi(K)\}$ belonging to group $1$ that are weaker than the group $2$ user. It can be seen that these receivers do not learn any information on channel $i$.  Thus among all the set of {\emph{active}} users on any given channel, the group $2$ user is the weakest user. Therefore the associated codeword of the group $2$ user constitutes the cloud center. 
\end{remark}

\section{Converse}
\label{sec:conv}
We first show that there exists a choice of auxiliary variables $\rvu_i(j)$ that satisfy the Markov chain condition
\begin{align}
\label{eq:rl}
\rvu_i(j) \rightarrow \rvx_i(j) \rightarrow \rvy_{\pi(1),i}(j) \cdots \rvy_{\pi(l_i),i}(j) \rightarrow \rvz_i(j) \rightarrow  \rvy_{\pi(l_i+1),i}(j)  \cdots \rvy_{\pi(K),i}(j).
\end{align}
such that the rates $R_1$ and $R_2$ are upper bounded by
\begin{align}
nR_1 &\le \sum_{i=1}^M \sum_{j=1}^n I(\rvx_i(j); \rvy_{k,i}(j) | \rvu_i(j), \rvz_i(j))+ 2n\eps_n \label{eq:R-ub1}\\
n R_2 &\le \sum_{i=1}^M \sum_{j=1}^n I(\rvu_i(j); \rvz_i(j) | \rvy_{k,i}(j)) + 2n\eps_n \label{eq:R-ub2}
\end{align}
for each $k \in \{1,\ldots, K\}$.

In particular we show that the choice of $\rvu_i(j)$ is given by the following:
\begin{align}
\rvu_i(j) = \big\{\rvm_2, \bbZ_{\bsi}^n, \bbz_{i,j+1}^n, \bbz_{i}^{j-1} \big\}\label{eq:u-def}
\end{align}
where we introduce (c.f.~\eqref{eq:rl})
{\allowdisplaybreaks{
\begin{align}
\bbz^n_i &:= (\rvz_i^n, \rvy_{\pi(l_i+1),i}^n,\ldots, \rvy_{\pi(K),i}^n),\\
\bbZ^n_{\bsi} &:= (\bbz_1^n,\ldots, \bbz_{i-1}^n, \bbz_{i+1}^n, \ldots, \bbz_M^n),\\
\bbz_{i}^{j-1} &:= (\rvz_i^{j-1}, \rvy_{\pi(l_i+1),i}^{j-1},\ldots, \rvy_{\pi(K),i}^{j-1}),\\
\bbz_{i,j+1}^n &:=(\rvz_{i,j+1}^{n}, \rvy_{\pi(l_i+1),i,j+1}^{n},\ldots, \rvy_{\pi(K),i,j+1}^{n}),
\end{align}
}}and observe our choice of $\rvu_i(j)$ in~\eqref{eq:u-def}  indeed satisfies~\eqref{eq:rl}.
Note that $\bbz_i^n$ is the collection of the Group $2$ receiver's channel output as well as the output
of all the receivers $\{\pi(l_i+1),\ldots, \pi(K)\}$ in Group $1$ that are degraded with respect to the group $2$ receiver on channel $i$.

We begin with the secrecy constraint associated with message $\rvm_2$ with respect to user $k$ in group $1$. 
Let us define the following:
\begin{align}
\bar{\rvy}_{k,i}^n := \begin{cases}
\rvy_{k,i}^n, & \rvx_k \rightarrow \rvz_i \rightarrow \rvy_{k,i}\\
\rvz_i^n, & \rvx_k  \rightarrow \rvy_{k,i}\rightarrow \rvz_i,
\end{cases}\label{eq:by-def}
\end{align}
\begin{align}
&\bby_k^n :=(\bar{\rvy}_{k,1}^n,\ldots,\bar{\rvy}_{k,M}^n), \quad
\rvbz^n := (\rvz_1^n,\ldots, \rvz_M^n),\\
&\bby_{k,i}^n :=(\bar{\rvy}_{k,1}^n,\ldots,\bar{\rvy}_{k,i}^n),\quad
\rvbz_{i}^n:=(\rvz_1^n,\ldots, \rvz_{i}^n).
\end{align}
Thus $\bby_k^n$ corresponds to a weaker receiver, whose  output on channel $i$ is degraded to $\rvz_i^n,$
if user $k$ is stronger than the group $2$ user on this sub-channel.
Clearly we  have that $\frac{1}{n}I(\rvm_2; \bby_k^n) \le \eps_n$ whenever $\frac{1}{n}I(\rvm_2;\rvby_k^n) \le \eps_n$. We thus have
{\allowdisplaybreaks{\begin{align}
n(R_2-2\eps_n) &\le I(\rvm_2; \rvbz^n) - I(\rvm_2; \bby_k^n)\\
&\le I(\rvm_2; \rvbz^n | \bby_k^n)\\
&= \sum_{i=1}^M \sum_{j=1}^n I(\rvm_2; \rvz_i(j) | \rvz_i^{j-1}, \rvbz_{i-1}^n,\bby_k^n )\\
&\le \sum_{i=1}^M \sum_{j=1}^n I(\rvm_2,  \rvz_i^{j-1}, \rvz_{i, j+1}^n, \rvbz_{i-1}^n, \bby_{k\bsi}^n, \bar{\rvy}_{k, i}^{j-1}, \bar{\rvy}_{k,i,j+1}^n; \rvz_i(j)   | \bar{\rvy}_{k,i}(j))\\
&\le \sum_{i=1}^M \sum_{j=1}^n I(\rvm_2, \bbZ_{\bsi}^n, \bbz_{i,j+1}^n, \bbz_{i}^{j-1}; \rvz_i(j) | \bar{\rvy}_{k,i}(j))\label{eq:subset-in}\\
&= \sum_{i=1}^M \sum_{j=1}^n I(\rvu_i(j); \rvz_i(j)| \bar{\rvy}_{k,i}(j) )\\
&=\sum_{i=1}^M \sum_{j=1}^n I(\rvu_i(j); \rvz_i(j)| {\rvy}_{k,i}(j) )\label{eq:reduction}
\end{align}}}
where~\eqref{eq:subset-in} follows from the fact that
\begin{align}
(\rvbz_{i-1}^n, \bby_{k\bsi}^n) \subseteq \bbZ_{\bsi}^n, \quad
(\rvz_i^{j-1}, \bar{\rvy}_{k, i}^{j-1}) \subseteq \bbz_{i}^{j-1},\quad
(\rvz_{i, j+1}^n,\bar{\rvy}_{k,i,j+1}^n) \subseteq \bbz_{i,j+1}^n,
\end{align}and~\eqref{eq:reduction} follows from the fact whenever $\rvy_{k,i}(j) \neq \bar{\rvy}_{k,i}(j) $ then $\rvz_i(j)$
is a degraded version of $\rvy_{k,i}(j)$  and from~\eqref{eq:by-def}, we have that
\begin{align}I(\rvu_i(j); \rvz_i(j)| {\rvy}_{k,i}(j) ) = I(\rvu_i(j); \rvz_i(j)| \bar{\rvy}_{k,i}(j) )=0.\end{align}
This establishes~\eqref{eq:R-ub2}. 

Next, we upper bound $R_1$ as follows:
{\allowdisplaybreaks{\begin{align}
n(R_1-2\eps_n) & \le I(\rvm_1; \rvby_k^n) - I(\rvm_1; \rvbz^n, \rvm_2) \\
&\le I(\rvm_1; \rvby_k^n | \rvbz^n, \rvm_2) \\
&\le \sum_{i=1}^M \sum_{j=1}^n I(\rvm_1; \rvy_{k,i}(j) | \rvy_{k,i}^{j-1}, \rvby_{k,i-1}^n, \rvbz^n, \rvm_2)\\
&\le \sum_{i=1}^M \sum_{j=1}^n H (\rvy_{k,i}(j) | \rvy_{k,i}^{j-1}, \rvby_{k,i-1}^n, \rvbz^n, \rvm_2) - H (\rvy_{k,i}(j) | \rvy_{k,i}^{j-1}, \rvby_{k,i-1}^n, \rvbz^n, \rvm_1,\rvm_2, \rvx_i(j)) \\
&= \sum_{i=1}^M \sum_{j=1}^n H (\rvy_{k,i}(j) | \rvy_{k,i}^{j-1}, \rvby_{k,i-1}^n, \rvbz^n, \rvm_2)  - H(\rvy_{k,i}(j)|\rvx_i(j), \rvz_i(j)) \label{eq:indep-x}\\
&\le \sum_{i=1}^M \sum_{j=1}^n H (\rvy_{k,i}(j) |  \rvbz^n, \rvm_2)  - H(\rvy_{k,i}(j)|\rvx_i(j), \rvz_i(j)) \label{eq:zero-t1}\\
&= \sum_{i=1}^M \sum_{j=1}^n H (\rvy_{k,i}(j) |  \bbZ_{\bsi}^n, \bbz_{i}^{j-1}, \bbz_{i,j+1}^n, \rvz_i(j), \rvm_2)  - H(\rvy_{k,i}(j)|\rvx_i(j), \rvz_i(j))
\label{eq:degraded-z}\\
&=\sum_{i=1}^M \sum_{j=1}^n H(\rvy_{k,i}(j) | \rvu_i(j),\rvz_i(j))  - H(\rvy_{k,i}(j)|\rvx_i(j), \rvz_i(j), \rvu_i(j))\\
&=\sum_{i=1}^M \sum_{j=1}^n I(\rvx_i(j); \rvy_{k,i}(j)|\rvu_i(j),\rvz_i(j)),
\end{align}}}
where~\eqref{eq:indep-x} follows from the fact that for our channel model $(\rvy_{k,i}(j), \rvz_i(j))$ are independent of all other random variables given $\rvx_i(j)$
whereas~\eqref{eq:degraded-z} follows from the fact that even though $\rvbz^n \subseteq \{\bbZ_{\bsi}^n, \bbz_{i}^{j-1}, \bbz_{i,j+1}^n, \rvz_i(j)\}$ holds,
the additional elements in the latter are only a degraded version of $\rvbz^n$. This establishes~\eqref{eq:R-ub1}.

To complete the converse,  let $\rvq_i$ to be a random variable uniformly distributed over the set $\{1,2,\ldots, n\}$ and furthermore we let $\rvu_i = (\rvu_i(\rvq_i), \rvq_i)$, $\rvx_i = \rvx_i(\rvq_i)$ etc.  Then~\eqref{eq:R-ub1} and~\eqref{eq:R-ub2} can be reduced to
\begin{align}
R_1-2\eps_n &\le \sum_{i=1}^M I(\rvx_i; \rvy_{k,i}|\rvu_i, \rvz_i, \rvq_i) =\sum_{i=1}^M I(\rvx_i; \rvy_{k,i}|\rvu_i, \rvz_i) \\
R_2-2\eps_n &\le \sum_{i=1}^M I(\rvu_i; \rvz_i | \rvy_{k,i}, \rvq_i) \le \sum_{i=1}^M I(\rvu_i; \rvz_i | \rvy_{k,i}).
\end{align}

The upper bound on the cardinality of $\cU_i$ follows by a straightforward application of Caratheodory's theorem and the proof is omitted.

\subsection{Special case of $K=2$ receivers}
For the case when there are $K=2$ receivers, the upper bound can be obtained via an alternative approach which involves first obtaining 
single-letter bounds for a particular genie-aided channel and then combining these bounds in a suitable manner. 

In particular, suppose that we only need to transmit message $\rvm_1$ to  receiver $1$ in group $1$ and that the message $\rvm_2$ only needs to be secure from user $2$ in group $1$. 
Under these relaxed constraints,  it can be shown that any achievable rate pair $(R_1,R_2)$ must satisfy:
\begin{align}
R_1 \le \sum_{i=1}^M I(\rvx_i; \rvy_{1,i}|\rvz_i, \rvu_i), \quad R_2 \le \sum_{i=1}^M I(\rvu_i; \rvz_i | \rvy_{2,i}), \label{eq:pr-1}
\end{align}
for some auxiliary variables $\{\rvu_i\}_{1\le i\le M}$ that satisfy the Markov chain in~\eqref{eq:rl}.
Similarly if we instead consider transmitting message $\rvm_1$ only to user $2$ in group $1$ and require secrecy of $\rvm_2$ only with respect to user $1$ in group $1$, it can be shown that any achievable rate pair $(R_1,R_2)$ must satisfy:
\begin{align}
R_1 \le \sum_{i=1}^M I(\rvx_i; \rvy_{2,i}|\rvz_i, \rvv_i), \quad R_2 \le \sum_{i=1}^M I(\rvv_i; \rvz_i | \rvy_{1,i}).\label{eq:pr-2}
\end{align}
for some auxiliary variables $\{\rvv_i\}_{1\le i \le M}$. Next, we show that on each sub-channel $i$ we can always set $\rvu_i=\rvv_i$ without affecting the upper bound. In particular we consider the following four cases:
\begin{itemize}
\item Group $2$ receiver satisfies $\rvx_i \rightarrow \rvz_i \rightarrow (\rvy_{1,i}, \rvy_{2,i})$: It suffices to take $\rvu_i = \rvv_i = \rvx_i$ in~\eqref{eq:pr-1} and~\eqref{eq:pr-2} as the contribution of this sub-channel in the expressions for $R_1$ is always zero.
\item Group $2$ receiver satisfies $\rvx_i \rightarrow (\rvy_{1,i}, \rvy_{2,i})\rightarrow \rvz_i $: It suffices to take $\rvu_i = \rvv_i = 0$ since the contribution of this sub-channel in the expressions for $R_2$ is zero.

\item Group $2$ receiver satisfies $\rvx_i \rightarrow \rvy_{1,i} \rightarrow \rvz_i \rightarrow \rvy_{2,i}$:~Since the contribution of sub-channel $i$  in the expressions of both $R_1$ and $R_2$ in ~\eqref{eq:pr-2} is zero, we can set $\rvv_i = \rvu_i$ without affecting the upper bound.
\item  Group $2$ receiver satisfies $\rvx_i \rightarrow \rvy_{2,i} \rightarrow \rvz_i \rightarrow \rvy_{1,i}$:~Since the contribution of sub-channel $i$  in the expressions of both $R_1$ and $R_2$ in~\eqref{eq:pr-1} is zero, we can  set $\rvu_i = \rvv_i$ without affecting the upper bound.
\end{itemize}
Thus we  need no more than one non-trivial auxiliary variable on each sub-channel. Setting $\rvv_i=\rvu_i$  in~\eqref{eq:pr-2} we have
\begin{align}
R_1 \le \sum_{i=1}^M I(\rvx_i; \rvy_{2,i}|\rvz_i, \rvu_i), \quad R_2 \le \sum_{i=1}^M I(\rvu_i; \rvz_i | \rvy_{1,i}).\label{eq:pr-3}
\end{align}
The converse follows by combining~\eqref{eq:pr-1} and~\eqref{eq:pr-3}.

Unfortunately when there are more than two receivers in group $1$, we have not been able to obtain the converse directly from such  single-letter expressions.
Therefore our approach in the previous section was to identify a single auxiliary random variable $\rvu_i$ as in~\eqref{eq:u-def} that is simultaneously compatible
with all the $n$-letter upper bound expressions. 

\section{Gaussian Channels}\label{sec:Gau}
In this section we provide a proof for Theorem~\ref{thm:Gauss}. Note that the achievability of the rate pairs $(R_1,R_2)$ constrained by \eqref{eq:Gau1} and \eqref{eq:Gau2} follows that of those constrained by \eqref{eq:R1e} and \eqref{eq:R2e} by setting $\rvx_i =\rvu_i+\rvv_i$, where $\rvu_i$ and $\rvv_i$ are independent $\mathcal{N}(0,P_i-Q_i)$ and $\mathcal{N}(0,Q_i)$ respectively for some $0 \leq Q_i \leq P_i$ and $i=1,\ldots,M$. For the rest of the section, we shall focus on proving the converse result.

Considering proof by contradiction, let us assume that $(R_1^o,R_2^o)$ is an achievable rate pair that lies \emph{outside} the rate region constrained by \eqref{eq:Gau1} and \eqref{eq:Gau2}.  Note that the maximum rate for message $\rvm_1$ is given by the right-hand side of \eqref{eq:Gau1} by setting $Q_i=P_i$ for all $i=1,\ldots,M$ \cite{ashishBC}, and the maximum rate for message $\rvm_2$ is given by the right-hand side of \eqref{eq:Gau2} by setting $Q_i=0$ for all $i=1,\ldots,M$\cite{LiuPrabhakaran:08,gopalaLaiElGamal:06Secrecy}. Thus, without loss of generality we may assume that $R_2^0=R_2^*+\delta$ for some $\delta >0$ where $R_2^*$ is given by
\begin{align}
\max_{(\mathbf{Q},R_2)} & \quad R_2\nonumber\\
\mbox{subject to} & \quad R_1^{o} \leq \sum_{i=1}^{M}A^{(1)}_{k,i}(\mathbf{Q}), \qquad \forall k=1,\ldots,K\label{eq:Cons1}\\
& \quad R_2 \leq \sum_{i=1}^{M}A^{(2)}_{k,i}(\mathbf{Q}),\qquad \forall k=1,\ldots,K\label{eq:Cons2}\\
& \quad Q_i \geq 0, \hspace{17pt} \forall i=1,\ldots,M\label{eq:Cons3}\\
& \quad Q_i \leq P_i, \hspace{13pt} \forall i=1,\ldots,M.\label{eq:Cons4}
\end{align}
For each $k=1,\ldots,K$ and $i=1,\ldots,M$ let $\alpha_k$, $\beta_k$, $M_{1,i}$ and $M_{2,i}$ be the Lagrangians that correspond to the constrains \eqref{eq:Cons1}--\eqref{eq:Cons4} respectively, and let
\begin{align}
L:= R_2+\sum_{k=1}^{K}\alpha_k\left[\sum_{i=1}^{M}A^{(1)}_{k,i}(\mathbf{Q})\!-\!R_1^o\!\right] \!+\! \sum_{k=1}^{K}\beta_k\left[\sum_{i=1}^{M}A^{(2)}_{k,i}(\mathbf{Q})\!-\!R_2\right]
 +\sum_{i=1}^{M}M_{1,i}Q_i+\sum_{i=1}^{M}M_{2,i}(P_i-Q_i).
\end{align}
It is straightforward to verify that the above optimization program that determines $R_2^*$ is a convex program. Therefore, taking partial derivatives of $L$ over $Q_i$, $i=1\ldots,M$ and $R_2$ respectively gives the following set of Karush-Kuhn-Tucker (KKT) conditions, which must be satisfied by any \emph{optimal} solution $(\mathbf{Q}^*,R_2^*)$:
\begin{align}
\sum_{k \in \mathcal{Y}_i}\alpha_k(Q_i^*+\sigma_{k,i}^2)^{-1}+\sum_{k \in \mathcal{Z}_i}\beta_k(Q_i^*+\sigma_{k,i}^2)^{-1}+M_{1,i} &= \left(\sum_{k \in \mathcal{Y}_i}\alpha_k+\sum_{k \in \mathcal{Z}_i}\beta_k\right)(Q_i^*+\delta_i^2)^{-1}+M_{2,i}\label{eq:KKT1}\\
\sum_{k=1}^{K}\beta_k & =1\label{eq:KKT2}\\
\alpha_k\left[\sum_{i=1}^MA^{(1)}_{k,i}(\mathbf{Q}^*)-R_1^o\right] &=0,\; \forall k=1,\ldots,K\label{eq:KKT3}\\
\beta_k\left[\sum_{i=1}^{M}A^{(2)}_{k,i}(\mathbf{Q}^*)-R_2^*\right] &=0,\; \forall k=1,\ldots,K\label{eq:KKT4}\\
M_{1,i}Q_i^* &=0, \; \forall i=1,\ldots,M\label{eq:KKT5}\\
M_{2,i}(P_i-Q_i^*) &=0, \; \forall i=1,\ldots,M\label{eq:KKT6}\\
\alpha_k, \beta_k & \geq 0,\; \forall k=1,\ldots,K\label{eq:KKT7}\\
M_{1,i}, M_{2,i} & \geq 0,\; \forall i=1,\ldots,M\label{eq:KKT8}
\end{align} 
where \begin{align}
\mathcal{Y}_i :=\{k: \sigma_{k,i}^2 < \delta_i^2\} \quad \mbox{and} \quad \mathcal{Z}_i :=\{k:\sigma_{k,i}^2 > \delta_i^2\}.
\end{align}
Note that $\delta>0$, so we have
\begin{align}
\left(\sum_{k=1}^{K}\alpha_k\right)R_1^o+R_2^o & > \left(\sum_{k=1}^{K}\alpha_k\right)R_1^o+R_2^*\\
&= \sum_{k=1}^{K}\left(\alpha_kR_1^o+\beta_kR_2^*\right)\label{eq:Contr100}\\
&= \sum_{k=1}^{K}\left[\alpha_k\sum_{i=1}^{M}A^{(1)}_{k,i}(\mathbf{Q}^*)+\beta_k\sum_{i=1}^{M}A^{(2)}_{k,i}(\mathbf{Q}^*)\right]\label{eq:Contr200}\\
&= \sum_{i=1}^{M}\sum_{k=1}^{K}\left[\alpha_kA^{(1)}_{k,i}(\mathbf{Q}^*)+\beta_kA^{(2)}_{k,i}(\mathbf{Q}^*)\right],\label{eq:Contr1}
\end{align}
where \eqref{eq:Contr100} follows from the KKT condition \eqref{eq:KKT2}, and \eqref{eq:Contr200} follows from the KKT conditions \eqref{eq:KKT3} and \eqref{eq:KKT4}.

Next, we shall show that by assumption $(R_1^o,R_2^o)$ is achievable, so we have
\begin{align}
\left(\sum_{k=1}^{K}\alpha_k\right)R_1^o+R_2^o &\leq \sum_{i=1}^{M}\sum_{k=1}^{K}\left[\alpha_kA^{(1)}_{k,i}(\mathbf{Q}^*)+\beta_kA^{(2)}_{k,i}(\mathbf{Q}^*)\right]\label{eq:Contr2}
\end{align}
which is an apparent contradiction to \eqref{eq:Contr1} and hence will help to complete the proof of the theorem.

To prove \eqref{eq:Contr2}, let us apply the converse part of Theorem~\ref{thm:DMC} on $(R_1^o,R_2^o)$ and write
\begin{align}
\left(\sum_{k=1}^{K}\alpha_k\right)R_1^o+R_2^o & \leq \left(\sum_{k=1}^{K}\alpha_k\right)\min_{1\le k \le K}\left\{\sum_{i=1}^M I(\rvx_i; \rvy_{k,i}| \rvu_i, \rvz_i)\right\}+ \min_{1\le k \le K}\left\{\sum_{i=1}^M I(\rvu_i; \rvz_i | \rvy_{k,i})\right\}\\
&\leq \sum_{k=1}^{K}\left[\alpha_k\sum_{i=1}^{M}I(\rvx_i; \rvy_{k,i}| \rvu_i, \rvz_i)\right]+ \sum_{k=1}^{K}\left[\beta_k\sum_{i=1}^{M} I(\rvu_i; \rvz_i | \rvy_{k,i})\right]\label{eq:Contr2.5}\\
& = \sum_{i=1}^{M}\sum_{k=1}^{K}\left[\alpha_kI(\rvx_i; \rvy_{k,i}| \rvu_i, \rvz_i)+\beta_k I(\rvu_i; \rvz_i | \rvy_{k,i})\right],\label{eq:Contr3}
\end{align}
where \eqref{eq:Contr2.5} follows from the well-known fact that minimum is no more than any weighted mean. By the degradedness assumption \eqref{eq:degOrd}, we have
\begin{align}
I(\rvx_i; \rvy_{k,i}| \rvu_i, \rvz_i) &=I(\rvx_i; \rvy_{k,i}| \rvu_i)-I(\rvx_i; \rvz_i| \rvu_i)\\
&=h(\rvy_{k,i}|\rvu_i)-h(\rvz_i|\rvu_i)-h(\rvn_{k,i})+h(\rvw_i)\\
&=h(\rvy_{k,i}|\rvu_i)-h(\rvz_i|\rvu_i)-\frac{1}{2}\log\left(\frac{\sigma_{k,i}^2}{\delta_i^2}\right)
\end{align} 
for any $k \in \mathcal{Y}_i$ and $I(\rvx_i; \rvy_{k,i}| \rvu_i, \rvz_i)=0$ for any $k \notin \mathcal{Y}_i$. Similarly, 
\begin{align}
I(\rvu_i; \rvz_i | \rvy_{k,i})& =I(\rvu_i; \rvz_i)-(\rvu_i; \rvy_{k,i})\\
&=h(\rvz_i)-h(\rvy_{k,i})-h(\rvz_i|\rvu_i)+h(\rvy_{k,i}|\rvu_i)\\
&\leq \frac{1}{2}\log\left(\frac{P_i+\delta_i^2}{P_i+\sigma_{k,i}^2}\right)-h(\rvz_i|\rvu_i)+h(\rvy_{k,i}|\rvu_i)\label{eq:WANL}
\end{align} 
for any $k \in \mathcal{Z}_i$, where \eqref{eq:WANL} follows from the worst additive noise Lemma \cite{diggaviCover01}, and 
$I(\rvu_i; \rvz_i | \rvy_{k,i})=0$ for any $k \notin \mathcal{Z}_i$. Thus, for each $i=1,\ldots,M$ we have
\begin{align}
\sum_{k=1}^{K} &\left[\alpha_kI(\rvx_i; \rvy_{k,i}| \rvu_i, \rvz_i)+\beta_kI(\rvu_i; \rvz_i | \rvy_{k,i})\right]\nonumber\\
&\leq \sum_{k \in \mathcal{Y}_i}\alpha_k\left[h(\rvy_{k,i}|\rvu_i)-h(\rvz_{i}|\rvu_i)-\frac{1}{2}\log\left(\frac{\sigma_{k,i}^2}{\delta_i^2}\right)\right]+\nonumber\\
& \hspace{20pt}\sum_{k \in \mathcal{Z}_i}\beta_k\left[\frac{1}{2}\log\left(\frac{P_i+\delta_i^2}{P_i+\sigma_{k,i}^2}\right)-h(\rvz_i|\rvu_i)+h(\rvy_{k,i}|\rvu_i)\right]\\
&= \sum_{k \in \mathcal{Y}_i}\alpha_kh(\rvy_{k,i}|\rvu_i)+\sum_{k \in \mathcal{Z}_i}\beta_kh(\rvy_{k,i}|\rvu_i)-\left(\sum_{k \in \mathcal{Y}_i}\alpha_k+\sum_{k \in \mathcal{Z}_i}\beta_k\right)h(\rvz_{i}|\rvu_i)-\nonumber\\
& \hspace{20pt} \sum_{k \in \mathcal{Y}_i}\frac{\alpha_k}{2}\log\left(\frac{\sigma_{k,i}^2}{\delta_i^2}\right)+\sum_{k \in \mathcal{Z}_i}\frac{\beta_k}{2}\log\left(\frac{P_i+\delta_i^2}{P_i+\sigma_{k,i}^2}\right).
\label{eq:Contr300}
\end{align}
We have the following lemma, which is the scalar version of the extremal inequality established in \cite[Theorem~2]{liu-vector}.
\begin{lemma}
For any real scalars $\alpha_k$, $\beta_k$, $Q_i^*$, $M_{1,i}$ and $M_{2,i}$ that satisfy KKT conditions \eqref{eq:KKT1} and \eqref{eq:KKT5}--\eqref{eq:KKT8}, we have
\begin{align}
\sum_{k \in \mathcal{Y}_i}&\alpha_kh(\rvy_{k,i}|\rvu_i)+\sum_{k \in \mathcal{Z}_i}\beta_kh(\rvy_{k,i}|\rvu_i)- \left(\sum_{k \in \mathcal{Y}_i}\alpha_k+\sum_{k \in \mathcal{Z}_i}\beta_k\right)h(\rvz_{i}|\rvu_i)  \nonumber\\
& \leq\sum_{k \in \mathcal{Y}_i}\frac{\alpha_k}{2}\log(Q_i^*+\sigma_{k,i}^2)+\sum_{k \in \mathcal{Z}_i}\frac{\beta_k}{2}\log(Q_i^*+\sigma_{k,i}^2)-
 \frac{\sum_{k \in \mathcal{Y}_i}\alpha_k+\sum_{k \in \mathcal{Z}_i}\beta_k}{2}\log(Q_i^*+\delta_{i}^2)
\label{eq:ExtIneq}
\end{align}
for any $(\rvu_i,\rvx_i)$ that is independent of the additive Gaussian noise $(\rvn_{1,i},\ldots,\rvn_{K,i},\rvw_i)$ and such that $E[\rvx_i^2]\leq P_i$.
\hfill$\Box$
\end{lemma}

We note here that the extremal inequality in \cite[Theorem~2]{liu-vector} was established using a \emph{vector} generalization of Costa's entropy-power inequality. The scalar version that we used here, however, can be directly established using the \emph{original} Costa's entropy-power inequality \cite{costa}. Substituting \eqref{eq:ExtIneq} into \eqref{eq:Contr300} gives 
\begin{align}
\sum_{k=1}^{K} &\left[\alpha_kI(\rvx_i; \rvy_{k,i}| \rvu_i, \rvz_i)+\beta_kI(\rvu_i; \rvz_i | \rvy_{k,i})\right]\nonumber\\
&\leq \sum_{k \in \mathcal{Y}_i}\frac{\alpha_k}{2}\log(Q_i^*+\sigma_{k,i}^2)+\sum_{k \in \mathcal{Z}_i}\frac{\beta_k}{2}\log(Q_i^*+\sigma_{k,i}^2)- \frac{\sum_{k \in \mathcal{Y}_i}\alpha_k+\sum_{k \in \mathcal{Z}_i}\beta_k}{2}\log(Q_i^*+\delta_{i}^2)-\nonumber\\
& \hspace{20pt} \sum_{k \in \mathcal{Y}_i}\frac{\alpha_k}{2}\log\left(\frac{\sigma_{k,i}^2}{\delta_i^2}\right)+\sum_{k \in \mathcal{Z}_i}\frac{\beta_k}{2}\log\left(\frac{P_i+\delta_i^2}{P_i+\sigma_{k,i}^2}\right)\\
&= \sum_{k \in \mathcal{Y}_i}\alpha_k\left[\frac{1}{2}\log\left(\frac{Q^*_i+\sigma_{k,i}^2}{\sigma_{k,i}^2}\right)-\frac{1}{2}\log\left(\frac{Q^*_i+\delta_{i}^2}{\delta_{i}^2}\right)\right]+\notag\\
&\hspace{20pt} \sum_{k \in \mathcal{Z}_i}\beta_k\left[\frac{1}{2}\log\left(\frac{P_i+\delta_i^2}{Q^*_i+\delta_{i}^2}\right)-\frac{1}{2}\log\left(\frac{P_i+\sigma_{k,i}^2}{Q^*_i+\sigma_{k,i}^2}\right)\right]\\
& =\sum_{k=1}^{K}\left[\alpha_kA^{(1)}_{k,i}(\mathbf{Q}^*)+\beta_kA^{(2)}_{k,i}(\mathbf{Q}^*)\right].\label{eq:Contr400}
\end{align}
Further substituting \eqref{eq:Contr400} into \eqref{eq:Contr3} completes the proof of \eqref{eq:Contr2}. We have thus completed the proof of Theorem~\ref{thm:Gauss}.

\section{Fading Channels}
\label{sec:fading}

To establish the connection to fading channels, first observe that Theorem~\ref{thm:Gauss} and Corollary~\ref{corol:Gauss} can be extended in the following way. Consider the following scalar Gaussian broadcast channel with $K+1$ users:
\begin{align}
\rvy_k(t) &= \rvx(t) + \rvn_k(t)\\
\rvz(t) &= \rvx(t) + \rvw(t), \quad t=1,\ldots,T.\label{eq:Gau100}
\end{align}
At each time sample $t$, the additive noise $(\rvn_1(t),\ldots,\rvn_K(t),\rvw(t))$ are independent zero-mean Gaussian with the variances $(\sigma_1^2,\ldots,\sigma_K^2,\delta^2)$ selected at random as $(\sigma_{1,i}^2,\ldots,\sigma_{K,i}^2,\delta_i^2)$ with probability $p_i$, $i=1,\ldots,M$. Both the selection of the noise variances and the realization of the additive noise are assumed to be independent across the time index $t$ and revealed to all the terminals. We are interested in the ergodic scenario where the duration $T$ of communication can be arbitrarily large. The following extension of Thoerem~\ref{thm:Gauss}
readily follows and its proof will be omitted.
\begin{corol}
For the scalar Gaussian broadcast channel considered above, the capacity region consists of all rate pairs $(R_1, R_2)$ that satisfy
\begin{align}
R_1 \!&\leq\! \min_{1 \leq k \leq K} \sum_{i=1}^{M}p_i\left[\frac{1}{2}\log\left(\frac{Q_i+\sigma_{k,i}^2}{\sigma_{k,i}^2}\right)-
\frac{1}{2}\log\left(\frac{Q_i+\delta_{i}^2}{\delta_{i}^2}\right)\right]^+\label{eq:prob-chan-1}\\
R_2\! &\leq\! \min_{1 \leq k \leq K} \sum_{i=1}^{M}p_i\left[\frac{1}{2}\log\left(\frac{P_i+\delta_{i}^2}{Q_i+\delta_{i}^2}\right)-
\frac{1}{2}\log\left(\frac{P_i+\sigma_{k,i}^2}{Q_i+\sigma_{k,i}^2}\right)\right]^+\label{eq:prob-chan-2}
\end{align}
for some $0 \leq Q_i \leq P_i$ and $i=1,\ldots,M$. \hfill$\Box$
\label{corol:prob-chan}
\end{corol}

Clearly if the fading coefficients in~\eqref{eq:fading} are all discrete-valued, then the result in Theorem~\ref{thm:fading} follows immediately from Corollary~\ref{corol:prob-chan}. When the fading coefficients are continuous valued, we can generalize Theorem~\ref{thm:Gauss} by suitably quantizing the channel gains.

First  without loss of generality, we  assume that each fading coefficient is real-valued, since each receiver can cancel out the phase
of the fading gain through a suitable  multiplication at the receiver.  Consider a discrete set
$$\cA := \{A_1, A_2,\ldots, A_N, A_{N+1}\}$$
where $A_i \le A_{i+1},$ $A_1:=0$, $A_N := J$ and $A_{N+1}:= \infty$ holds. 

Given a set of channel gains $(\rvh_1(i),\ldots, \rvh_K(i), \rvg(i))$ in coherence block $i$, we discretize them to one of $(N+1)^{K+1}$ states as described below. 
\begin{itemize}
\item Encoding message $\rvm_1$: Suppose that the channel gain of receiver $k$ satisfies $A_q \le \rvh_k(i) \le A_{q+1}$, then we assume that the  channel gain equals $\rvs_{i,k} =A_q$. If the channel gain of the group $2$ user satisfies $A_q \le \rvg(i) \le A_{q+1}$ 
then we assume that its channel gain equals $\rvsu_{i,K+1} = A_{q+1}$. 
\item Encoding message $\rvm_2$: Suppose that the channel gain of the group $2$ receiver satisfies $A_q \le \rvg(i) \le A_{q+1}$, then we assume that the channel gain equals $\rvs_{K+1} = A_q$. If the channel gain of a group $1$ receiver satisfies $A_q \le \rvh_k(i) \le A_{q+1}$ then we assume it equals $\rvsu_k = A_{q+1}$.
\end{itemize}

Thus the channel gains in coherence block are mapped to one of $L = (N+1)^{K+1}$ states $\{\rvbs_j\}_{j=1}^L$. We denote the channel gains of the associated receivers in state $\rvbs_j$ as $(\rvs_{j,1},\ldots, \rvs_{j,K}, \rvs_{j,K+1})$ and the channel gains of the associated eavesdroppers as $(\rvsu_{j,1},\ldots, \rvsu_{j,K+1})$. Note that in our notation, the $K$ receivers in group $1$ are labeled $\{1,\ldots, K\}$ while the group $2$ receiver is labeled $\{K+1\}$.

With the above quantization procedure it suffices to consider a coding scheme associated for  $L =(N+1)^{K+1}$
parallel channels, where each parallel channel corresponds to one state realization $\rvbs_j$.
Using Corollary~\ref{corol:prob-chan} the following rate pair $(R_1, R_2)$ is achievable:
\begin{align}
R_1 &\le \min_{1\le k \le K} \sum_{j=1}^L \Pr(\rvbs_j)A^{(1)}_{j,k}(\rvbs_j) \label{eq:fading-R1-achiev}\\
R_2 &\le \min_{1\le k \le K} \sum_{j=1}^L \Pr(\rvbs_j) A^{(2)}_{j,k}(\rvbs_j),\label{eq:fading-R2-achiev}
\end{align}
where 
{\allowdisplaybreaks{\begin{align}
A^{(1)}_{j,k}(\rvbs_j) &:= \bigg\{ \log \frac{1 + Q(\rvbs_j)|\rvs_{j,k}|^2}{1+ Q(\rvbs_j)|\rvsu_{j,K+1}|^2}\bigg\}^+ \label{eq:A1-def}\\
A^{(2)}_{j,k}(\rvbs_j) &:=\bigg\{ \!\log \frac{1 + P(\rvbs_j)|\rvs_{j,K+1}|^2}{1+ Q(\rvbs_j)|\rvs_{j,K+1}|^2} \!-\!
\log \frac{1 + P(\rvbs_j)|\rvsu_{j,k}|^2}{1+ Q(\rvbs_j)|\rvsu_{j,k}|^2}  \!\bigg\}^+.
\end{align}}}

For any $J$, taking the limit $N \rightarrow \infty$ we have that
\begin{align}
\sum_{j=1}^L \Pr(\rvbs_j)A^{(1)}_{j,k}(\bh,g) &\rightarrow \oint_0^J \int_0^J A^{(1)}_{k}(\bh, g) dF(g) dF(\bh) \\
&=\oint_0^J \int_0^\infty A^{(1)}_{k}(\bh,g)  dF(g) dF(\bh) \label{eq:int-limit}
\end{align}
where 
$$A^{(1)}_{k}(\bh,g) = \bigg\{\log\frac{1 + Q(\bh,g)|h_k|^2}{1+Q(\bh,g)|g|^2}\bigg\}^+,$$
and~\eqref{eq:int-limit} follows from the fact that $A_k^{(1)}(\cdot)=0$ for $\rvsu_{K+1} > J$. Finally, by taking $J$ arbitrarily large, the right hand side in~\eqref{eq:fading-R1-achiev} approaches
\begin{align}
R_1 \le \min_{1\le k \le K}  \oint_0^\infty \int_0^\infty A^{(1)}_{k}(\bh,g)  dF(g)  dF(\bh) 
\end{align}
as required. In a similar fashion the achievability of $R_2$ can be established.

The converse follows by noticing that if the channel gains are revealed non-causally to the terminals, the system reduces to a parallel channel model and the result in Theorem~\ref{thm:Gauss} immediately applies.

\subsection{Numerical Results}
\begin{figure}
\begin{center}
\includegraphics[scale=0.3]{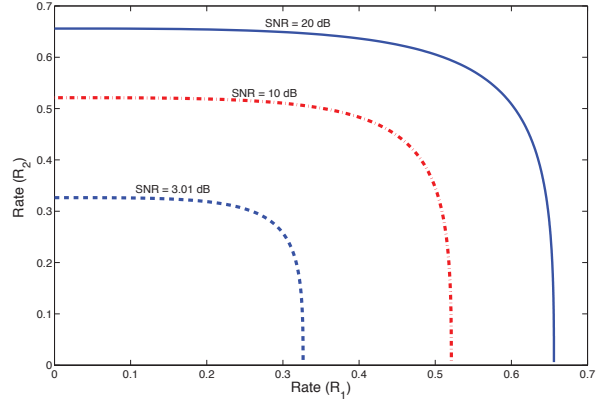}
\caption{Achievable rates (nats/symbol)  for the two groups at different SNR values. The x-axis shows the rate $R_1$ for group $1$ whereas the y-axis shows the rate $R_2$ for group $2$.}
\label{fig:fading}
\end{center}
\end{figure}

In order to evaluate the achievable rate region, we assume that the fading gains are all  sampled from $\CN(0,1)$. Furthermore instead of finding the optimal power allocation we assume a potentially sub-optimal power allocation:
\begin{align}
Q(\bh, g) = \begin{cases}
P, & |g|^2 \ge \theta\\
0, & |g|^2 < \theta.
\end{cases}\label{eq:theta-select}
\end{align}
where $\theta$ is a certain fixed parameter and assume that $P(\bh,g) = P$ for all values of $(\bh,g)$. Notice that our power allocation does not depend on the channel gains of the receivers in group $1$.This is a reasonable simplification when $K$ is large and the channel gains $(\rvh_1,\ldots, \rvh_K)$ are identically distributed. The achievable rate expressions~\eqref{eq:R1-fading}
and~\eqref{eq:R2-fading} reduce to:
\begin{align}
R_1 &\le \Pr(|\rvg|^2 \le \theta)E\left[ \bigg\{\log\frac{1+P|\rvh|^2}{1+P|\rvg|^2}\bigg\}^+ \bigg| |\rvg|^2 \le \theta \right] \label{eq:R-achiev-1}\\
R_2 &\le \Pr(|\rvg|^2 \ge \theta)  E\bigg[ \bigg\{\log\frac{1+P|\rvg|^2}{1+P|\rvh|^2}\bigg\}^+ \bigg| |\rvg|^2 \ge \theta \bigg]\label{eq:R-achiev-2}
\end{align}

In Fig.~\ref{fig:fading}, we plot the achievable rates for $P\in \{2,10,100\}$.   We make the following observations:
\begin{itemize}
\item The corner points for $R_1$ and $R_2$ are obtained by setting $\theta = \infty$ and $\theta=0$ respectively.  By symmetry of the rate expressions in~\eqref{eq:R-achiev-1} and~\eqref{eq:R-achiev-2}, it is clear that both the corner points  evaluate to the same numerical constant. 
\item As we approach the corner point $(0,R_2)$ the boundary of the capacity region is nearly flat. Any coherence block, where $|\rvg(i)| \le \min_{1\le k \le K} |\rvh_k(i)|$ is clearly not useful to the group $2$ receiver. By transmitting $\rvm_1$ in these slots one can increase the rate $R_1$ without decreasing $R_2$. 
\item As we approach the corner point $(R_1, 0),$ the boundary of the capacity region is nearly vertical. The argument is very similar to the previous case. In any period where  $|\rvg(i)| \ge \max_{1\le k \le K} |\rvh_k(i)|$ one cannot transmit to group $1$. By transmitting $\rvm_2$ in these slots we increase $R_2$ without decreasing $R_1$.
\item We observe that a natural alternative to the proposed scheme is time-sharing. The rate achieved by such a scheme corresponds to a straight line connecting the corner points. The rate-loss associated with such a scheme is significant compared to the proposed scheme. 
\end{itemize}

\section{Conclusions}
We establish the optimality of a superposition construction for private broadcasting of two messages to two groups of receivers over independent parallel channels, when there are an arbitrary number of receivers in group $1$ but there is only one receiver in group $2$. We observe that in the optimal construction the codewords of group $2$ must constitute the ``cloud centers" whereas the codewords of group $1$ must constitute the ``satellite codewords".  For the case of Gaussian sub-channels the optimality of Gaussian codebooks is established. This is accomplished by obtaining a Lagrangian dual for each point on the boundary of the capacity region and then using an extremal inequality to show that the resulting expression is maximized using a Gaussian input distribution. An extension to block-fading channels is also discussed. Numerical results for Rayleigh-fading channels indicate that the proposed scheme can provide significant performance gains over naive time-sharing techniques.

\bibliographystyle{IEEEtran}
\bibliography{sm}

\end{document}